\documentclass[a4paper]{article}

\usepackage{arydshln}
\usepackage{graphicx}
\usepackage{amsmath}
\usepackage{amssymb, verbatim}
\usepackage{mathrsfs}
\usepackage{dsfont}
\usepackage{url}
\usepackage[usenames,dvipsnames]{color}
\usepackage[compress]{cite}
\usepackage{mdwlist}
\usepackage{paralist}
\usepackage{algorithmic}
\usepackage{algorithm}
\usepackage{stmaryrd}
\usepackage[hyphenbreaks]{breakurl}

\usepackage[font=footnotesize]{caption}

\newtheorem{lemma}{\sc Lemma}
\newtheorem{theorem}[lemma]{\sc Theorem}
\newtheorem{proposition}[lemma]{\sc Proposition}
\newtheorem{corollary}[lemma]{\sc Corollary}
\newtheorem{remark}{\sc Remark}
\newtheorem{assumption}{\sc Assumption}
\newtheorem{definition}{\sc Definition}

\newcommand{\citep}{\cite }

\DeclareMathOperator*{\argmin}{arg\,min}
\DeclareMathOperator*{\argmax}{arg\,max}

\newcommand{\N}{\llbracket N\rrbracket}
\newcommand{\M}{\llbracket M\rrbracket}
\newcommand{\R}{\llbracket R\rrbracket}

\newenvironment{proof}{{\noindent \bf Proof:\ }}{ \hfill $\square$}

\begin{document}

\date{}

\title{A Study of Truck Platooning Incentives \\ Using a Congestion Game\thanks{ An early version of this paper on motivating the modeling assumptions and to extract appropriate simulation parameters using real traffic data was presented at the 16th International IEEE Conference on Intelligent Transportation Systems (ITSC~2013)~\cite{FarokhiITSC2013}. A preliminary version of the theoretical results is submitted for presentation~\cite{FarokhiSubmitted2013}. The work was supported by the Swedish Research Council, the Knut and Alice Wallenberg Foundation, and the iQFleet project.}}
\author{Farhad~Farokhi$^\dag$~and~Karl~H.~Johansson\thanks{The authors are with ACCESS Linnaeus Center, School of Electrical Engineering, KTH Royal Institute of Technology, SE-100 44 Stockholm, Sweden. E-mails: \{farakhi,kallej\}@kth.se }}

\maketitle

\begin{abstract} We introduce an atomic congestion game with two types of agents, cars and trucks, to model the traffic flow on a road over various time intervals of the day. Cars maximize their utility by finding a trade-off between the time they choose to use the road, the average velocity of the flow at that time, and the dynamic congestion tax that they pay for using the road. In addition to these terms, the trucks have an incentive for using the road at the same time as their peers because they have platooning capabilities, which allow them to save fuel. The dynamics and equilibria of this game-theoretic model for the interaction between car traffic and truck platooning incentives are investigated. We use traffic data from Stockholm to validate parts of the modeling assumptions and extract reasonable parameters for the simulations. We use joint strategy fictitious play and average strategy fictitious play to learn a pure strategy Nash equilibrium of this game. We perform a comprehensive simulation study to understand the influence of various factors, such as the drivers' value of time and the percentage of the trucks that are equipped with platooning devices, on the properties of the Nash equilibrium. 
\end{abstract}

\section{Introduction}
\subsection{Motivation}
Urban traffic congestion creates many problems, such as increased transportation delays and fuel consumption, air pollution, and dampened economic growth in heavily congested areas~\cite{barth2008real,Hymel2009127,fuglestvedt2008climate}. A recent study~\citep{fuglestvedt2008climate} shows that the transportation has contributed to approximately 15\% of the total man-made carbon-dioxide since preindustrial era and suggests that it will be responsible for roughly 16\% of the carbon-emission over the next century. To circumvent part of these issues, the local governments in some urban areas introduced congestion taxes to manage the traffic congestion over existing infrastructures. For instance, Stockholm implemented a congestion taxing system in August, 2007 after a seven-month trial period in 2006. A survey of the influence of the congestion taxes over the trial period can be found in~\cite{Eliasson2009240}, which  shows significant improvements in travel times as well as favorable economic and environmental effects. Behavioral aspects and other influences of the Stockholm congestion taxing system is discussed in~\cite{Karlstrom2009283,WinslottHiselius2009269,Eliasson2006602,Borjesson20121}.

In parallel to reducing the congestion, we can employ other means to improve the fuel efficiency and decrease the carbon emission~\cite{barth2008real}. One way to improve the fuel efficiency of vehicles is platooning, as vehicles experience a reduced air drag force when they travel in platoons~\cite{AlAlam10,Xiaoliang6426543,zabat1995estimates,BONNETFRITZ,1603555}. Trucks or heavy-duty vehicles can significantly improve their fuel efficiency by platooning with their peers. In~\cite{AlAlam10}, the authors report 4.7\%-7.7\% reduction in the fuel consumption (depending on the distance between the vehicles among other factors) when two identical trucks move close to each other at $70\,\mathrm{km/h}$. In a futuristic scenario when several trucks are equipped with platooning devices, they are able to save fuel by cooperating with each other. However, implementing truck platooning in a large-scale setup is not easy since a global decision-maker might become complex and the vehicles can belong to competing entities. In addition, it is interesting to study if a desirable behavior can emerge from simple local strategies. In this paper, we consider such a case where the traffic flow can be modeled as a congestion game and the desired behavior corresponds to an equilibrium of this game.

\subsection{Related Studies}
Modeling the traffic flow using congestion games or routing games is a well-known problem~\cite{Fisk1984301,Levinson2005691,correa2005inefficiency,rosenthal1973class,
krichenestackelberg,braess2005paradox,awerbuch2005price,wardrop1952road,Yang2007841}. Rosenthal~\cite{rosenthal1973class} presented a noncooperative game in which a finite number of players compete for using a finite set of resources with application to modeling transport networks. He showed that a class of these games admit at least one pure strategy Nash equilibrium (an action profile in which no agent has an incentive to unilaterally deviate from her action). Later, the authors of~\cite{monderer1996potential} showed that atomic congestion games are indeed potential games (i.e., there exists a potential function, such that its variation when only one agent changes her action is equal to the variation of the utility of the corresponding agent) under some conditions and, hence, one can find a Nash equilibrium by minimizing the potential function. For a survey of these and related results, see~\cite{voorneveld1999congestion}. Most of these studies modeled the \emph{route selection} using an atomic congestion game. Recently, the authors of~\cite{gameroad} utilized a congestion game for modeling instead the time interval in which drivers decide to use a road.

This setup may be extended to weighted congestion games in which every agent is associated with a (splittable or unsplittable) demand (not equal and more than a single unit) that should be routed over the network. In~\cite{NET:NET3230030104}, Rosenthal showed that a Nash equilibrium does not necessarily exist in these games if the agents can split their demand. The authors of~\cite{libman2001atomic,fotakis2005selfish,goemans2005sink} constructed counterexamples to show that a Nash equilibrium does not necessarily exist also for unsplittable demands as well. However, when cost functions (i.e., latencies) of each road are affine functions, an equilibrium certainly exists (and may be found in pseudo-polynomial time)~\cite{fotakis2005selfish}. In~\cite{libman2001atomic}, it was also proved that an equilibrium may exist for a special class of cost functions (that are only a function of the residual capacity on each edge) on parallel networks. The largest class of latency functions for which the game admits an equilibrium were explored in~\cite{harks2012existence}. It was also shown that a weighted congestion game admits an exact potential function (a weighted potential function) if and only if the set of costs contains only affine functions (affine or exponential functions)~\cite{harks2011characterizing}.

The studies discussed above mainly consider homogeneous congestion games in which all the drivers on a road at any given time interval perceive the same cost function (e.g, the drivers only consider the latency in their decision-making and they all have the same sensitivity to the latency as well). However, in road traffic networks, this assumption might not be realistic. For instance, as we will see in this paper, whenever the drivers include the fuel consumption in their decision making, trucks and cars potentially have different cost functions even if they observe the same latency when using the road. To capture this phenomenon, we extend the model in~\cite{gameroad} to an atomic congestion game with two types of agents, namely, cars and trucks. 
Notice that the problem of heterogeneous congestion and routing games have been studied extensively in the past~\cite{milchtaich1996congestion,dafermos1972traffic, netter1971equilibrium}. For instance, in~\cite{milchtaich1996congestion}, the author formulated a congestion game in which each player has a specific cost function that depends on the congestion. In that study, it was shown that every unweighted congestion game with player-specific cost functions admits at least one equilibrium; however, this results may not be generalized to weighted congestion games with player-specific cost functions in general. In addition, generally, even unweighted congestion games with player-specific cost functions do not admit a potential function. For routing games, in which a continuum of players route an infinitesimal amount of flow, it was proved that a potential function exists if a symmetry condition is satisfied for the cost functions (i.e., various classes of agents bother or delight each other equally)~\cite{dafermos1972traffic,engelson2006congestion}. A class of necessary and sufficient conditions for the existence of potential functions was presented in~\cite{FarokhiAllerton2013}. Conditions for the (essential) uniqueness of the equilibrium in multi-class routing games were also presented in~\cite{konishi2004uniqueness,daganzo1983stochastic}. 

Motivated by the fact that the Nash equilibrium is generally inefficient, the price of anarchy (i.e., the worst-case ratio of the social welfare function for a Nash equilibrium over the social welfare function for a socially optimal solution) of atomic congestion games with linear latency functions was studied in~\cite{christodoulou2005price}. Several studies have proposed congestion taxes (also known as tolls) to improve the social cost function when all the agents are equally sensitive to the proposed taxes~\cite{marden2009joint,marden2009payoff,Pigou1,roughgarden2007routing} as well as when they have different sensitivities~\cite{zhang2012enhancing,zhang2008multiclass,zhang2011competitive,yang2004multi}. For instance, in~\cite{yang2004multi}, tolls were introduced to minimize the total travel time and the total travel cost (as a bi-objective optimization problem). This setup was generalized in~\cite{zhang2008multiclass} to also admit  entities that own several agents (and wish to optimize the combined utility of those agents). The idea of maximizing the reserve capacity of the network was approached in~\cite{zhang2012enhancing}. A scenario in which the network is managed by several decision-makers (with conflicting objectives) across various regions was discussed in~\cite{zhang2011competitive}. The authors of~\cite{marden2009joint,marden2009payoff,gameroad} presented congestion taxes so that the underlying congestion game admits the social welfare as a potential function. This is certainly of interest because it guarantees that the socially optimal decision is also a Nash equilibrium. However, in those studies, the authors needed to introduce a congestion tax for all the agents (and not only a subset of them).

\subsection{Contributions}
In this paper, we model the traffic flow at non-overlapping intervals of the day using an atomic\footnote{We use the term atomic to emphasize the fact that we are not dealing with a continuum of players or fractional flows when modeling the traffic flow as a congestion game~\cite{schmeidler1973equilibrium,6426543}.} congestion game with two types of agents. The agents of the first type are cars as well as trucks that do not have platooning equipments. For the sake of brevity, we call all these agents cars. They optimize their utility, which is a sum of the penalty for deviating from their preferred time for using the road, the average velocity of the traffic flow along the road, and the congestion tax that they pay for using the road at that time interval. The  agents of the second type are trucks equipped with platooning devices. For the sake of brevity, we call these agents trucks. In addition to the above mentioned terms, they have an incentive for using the road with other trucks (due to an increased chance for platooning and, hence, reducing their fuel consumption). 

We model the average velocity of the flow at each time interval as an affine function of the number of the vehicles that are using the road at that time interval. We use real traffic data from the northbound E4 highway from Lilla Essingen to the end of Fredh\"{a}llstunneln in Stockholm to validate this modeling assumption. 

We determine a necessary condition for the existence of a potential function for the introduced atomic congestion game with two types of agents and use this condition to prove that in general the congestion game is not a potential game. Therefore, we devise appropriate congestion taxes (specifically, a congestion taxing policy for cars and a platooning subsidy for trucks) to guarantee the existence of a potential function. Based on this result, we prove that the atomic congestion game admits at least one pure strategy Nash equilibrium under the proposed congestion tax--subsidy policy. Equipped with these results, we use joint strategy fictitious play and average strategy fictitious play to learn a Nash equilibrium. Intuitively, we interpret the learning algorithm as the way drivers decide on a daily basis to choose the time interval on which they are using the road by optimizing their utility given the history of their actions. Iterating over days, the drivers' decisions (i.e., the profile of the learning algorithm) converges almost surely to a pure strategy Nash equilibrium. Note that the potential games are certainly not the only classes of games for which variants of the fictitious play (e.g., joint strategy fictitious play) may converge to an equilibrium. To mention a few example, the authors of~\cite{Monderer1996258,monderer1996potential} introduced ordinal potential games and weighted potential games as two families of games for which the fictitious play converges in beliefs to a mixed strategy Nash equilibrium. For (generalized) ordinal potential games, one may also deduce the convergence of the joint strategy fictitious play to a pure strategy Nash equilibrium with probability one~\cite{marden2009joint}. These families of games are certainly more general than (exact) potential games. In this paper, as a starting point, we present necessary conditions for the existence of (exact) potential functions as well as imposing congestion taxes for guaranteeing the existence of such functions. Although conservative, this approach perhaps can be justified in the introduced problem due to the existence of intuitive taxing and subsidy policies (see Subsection~\ref{subsec:taxforpotential}).  A viable direction for future work is to investigate necessary and sufficient conditions so that a congestion game belongs to the category of ordinal or weighted potential games. In addition, as also mentioned earlier, congestion taxes were presented in~\cite{marden2009joint,marden2009payoff,gameroad} so that the congestion game admits the social welfare as a potential function. However, in contrast to the results of this paper, the authors of~\cite{marden2009joint,marden2009payoff} introduced a congestion tax for all the agents (and not only a subset of them) to improve the efficiency and considered homogeneous congestion games (with only one type of agents).

Finally, using the parameters extracted from the real congestion data, we construct a simulation setup to study the performance of the learning algorithms as well as the properties of the Nash equilibrium. For instance, we study the robustness to perturbations of the learning algorithm, e.g., accidents along the road, sudden weather changes, or temporary road constructions. We also consider the case when the drivers value their time differently, where the values are motivated by survey data from Stockholm area~\cite{Inregiareport}.

\subsection{Paper Organization}
The rest of the paper is organized as follows. In Section~\ref{sec:problemsetup}, we formulate the considered congestion game. We find a necessary condition of the existence of a potential function in Section~\ref{sec:potential}. In Sections~\ref{sec:JSFP} and~\ref{sec:ASFP}, we respectively introduce the joint strategy fictitious play and the average strategy fictitious play to learn a Nash equilibrium of the congestion game. Finally, we present the simulations in Section~\ref{sec:numericalexample} and conclude the paper in Section~\ref{sec:conclusion}.

\subsection{Notation}
Let $\mathbb{R}$, $\mathbb{Z}$, and $\mathbb{N}$ denote the sets of real, integer, and natural numbers, respectively. Furthermore, let $\mathbb{N}_0=\mathbb{N}\cup \{0\}$. We define $\N=\{1,\dots,N\}$ for any $N\in\mathbb{N}$. In this paper, all other sets are denoted by calligraphic letters such as $\mathcal{R}$. We use $|\mathcal{R}|$ to denote the cardinality of $\mathcal{R}$. Finally, we define the characteristic function $\mathbf{1}_{x=y}$ ($\mathbf{1}_{x\geq y}$) to be equal one whenever $x=y$ ($x\geq y$) holds true and to be equal to zero otherwise.
\section{Game-Theoretic Model} \label{sec:problemsetup}
We model the traffic flow at certain time intervals of the day on a given road using an atomic congestion game. The agents in this congestion game are the vehicles (or, rather the drivers of these vehicles) and their actions are the time intervals that they choose to use the road at each day. Let us divide the time of the day into $R\in\mathbb{N}$ non-overlapping intervals and denote each interval by $r_i$ for $i\in\R$. The set of all these intervals (i.e., agents' actions) is denoted by $\mathcal{R}=\{r_1,r_2,\dots,r_R\}$. We consider the case where the underlying congestion game is composed of two types of agents. As specified in the introduction, we name the agents of the first type cars and the agents of the second type trucks throughout the paper. We assume $N$ cars and $M$ trucks are playing in this congestion game and denote the actions of the cars and the trucks by $z=\{z_i\}_{i=1}^N$ and $x=\{x_i\}_{i=1}^M$, respectively. Let us describe the utilities of the cars and the trucks in the following subsections.

\subsection{Car Utility}
Car $i\in\N$ maximizes its utility given by
\begin{equation} \label{eqn:util:1}
U_i(z_i,z_{-i},x)=\xi_i^{\mathrm{c}}(z_i,T_i^{\mathrm{c}}) +v_{z_i}(z,x)+p_i^{\mathrm{c}}(z,x),
\end{equation}
where the mapping $\xi_i^{\mathrm{c}}:\mathcal{R}\times \mathcal{R} \rightarrow\mathbb{R}$ describes the penalty for deviating from the preferred time interval for using the road denoted by $T_i^{\mathrm{c}}\in\mathcal{R}$ (e.g., due to being late for work or delivering goods), $v_{z_i}(z,x)$ is the average velocity of the traffic flow at time interval $z_i$, and $p_i^{\mathrm{c}}(z,x)$ is a potential congestion tax for using the road on a specific time interval.

Following~\cite{gameroad,bell1997transportation,garavello2006traffic}, we assume that $v_r(z,x)$ (i.e., the average velocity at time interval $r\in\mathcal{R}$) is linearly dependent on the road congestion
\begin{equation} 
n_r(z,x)=\sum_{\ell=1}^N \mathbf{1}_{\{z_\ell=r\}}+\sum_{\ell=1}^M \mathbf{1}_{\{x_\ell=r\}},
\end{equation}
which is the total number of vehicles (both cars and trucks) that are using the road at $r\in\mathcal{R}$. 
Let us use real traffic data from sensors on the northbound E4 highway in Stockholm from Lilla Essingen to the end of Fredh\"{a}llstunneln (see Figure~\ref{figuremap}) to validate this assumption. The measurements are extracted during October 1--15, 2012. Figure~\ref{figure_validation} illustrates the average velocity of the flow as a function of the number of vehicles. As we can see, for up to 1000 vehicles, a linear relationship 
\begin{equation} \label{eqn:affineaverage}
v_r(z,x)=an_r(z,x)+b
\end{equation}
with $a=-0.0110$ and $b=84.9696$ describes the data well. However, for higher numbers of the vehicles, it fails to capture the behavior of around 20\% of the data (shown by the red dots in Figure~\ref{figure_validation}). Note that some of these outlier measurements can be caused by traffic accidents, sudden weather changes during the day, or temporary road constructions. A viable direction for future work is to introduce more complex velocity models in which the average velocity of the traffic flow may depend on the number of vehicles in the neighboring time intervals in addition to the current one. We may also need to  separate the effect of cars and trucks as one may expect heavier and larger vehicles to contribute more to the traffic congestion. However, in this paper, we use the simple model presented in~\eqref{eqn:affineaverage} and instead focus on platooning incentives.

\begin{figure}[!t]
\centering
\includegraphics[width=0.4\linewidth]{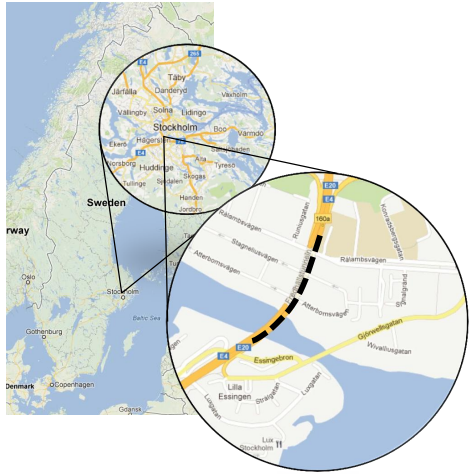}
\caption{\label{figuremap} The dashed black curve shows the segment of northbound E4 highway between Lilla Essingen and Fredh\"{a}llstunneln in Stockholm where we are using to validate the model and extract reasonable parameters. }
\end{figure}

\begin{figure}[!t]
\centering
\includegraphics[width=0.4\linewidth]{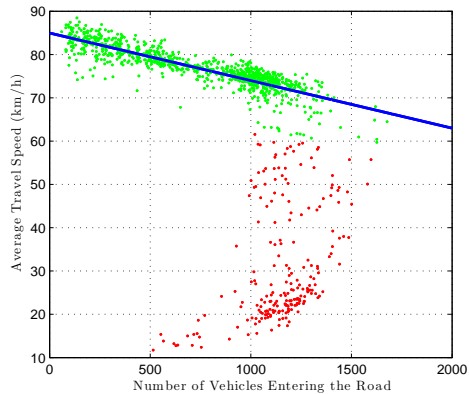}
\caption{\label{figure_validation} Average velocity of the traffic flow as a function of the number of vehicles that are entering the segment of northbound E4 highway between Lilla Essingen and Fredh\"{a}llstunneln for $15\,\mathrm{min}$ time intervals.}
\end{figure}

The choice of the penalty mappings $\xi_i^{\mathrm{c}}$, $i\in\N$, does not change the theoretical results presented in the paper, but it can capture various models of the drivers. For instance, following~\cite{gameroad}, we can use $\xi_i^{\mathrm{c}}(z_i,T_i^{\mathrm{c}})=\alpha_i^{\mathrm{c}} |z_i-T_i^{\mathrm{c}}|$, with scalar $\alpha_i^{\mathrm{c}}<0$, to describe the case where the driver of car $i$ is  penalized by deviating from the preferred time interval. With this function, the driver get penalized symmetrically no matter if she uses the road sooner or later than $T^\mathrm{c}_i$. By increasing $|\alpha_i^{\mathrm{c}}|$, she becomes less flexible. Another penalty function is $\xi_i^{\mathrm{c}}(z_i,T_i^{\mathrm{c}}) =\alpha_i^{\mathrm{c}} \max(z_i-T_i^{\mathrm{c}},0)$, which penalizes the driver of car $i$ only for being late. For the simulations in the paper, we assume that all vehicles use the first penalty mapping.

\subsection{Truck Utility}
Truck $j\in\M$ maximizes its utility given by
\begin{equation} \label{eqn:util:2}
\begin{split}
V_j(x_j,x_{-j},z)=&\,\xi_j^{\mathrm{t}}(x_j,T_j^{\mathrm{t}}) +v_{x_j}(z,x)+p_i^{\mathrm{t}}(z,x)+\beta v_{x_j}(z,x)g(m_{x_j}(x)),
\end{split}
\end{equation}
where, similar to the utilities of the cars, $\xi_j^{\mathrm{t}}(x_j,T_j^{\mathrm{t}})$ is the penalty for deviating from the preferred time $T_j^{\mathrm{t}}$ for using the road, $v_{x_j}(z,x)$ is the average velocity of the traffic flow, and $p_i^{\mathrm{t}}(z,x)$ is a potential congestion tax for using the road at time interval $x_j$. Trucks have an extra term $\beta v_{x_j}(z,x)g(m_{x_j}(x))$ in their utility because of their benefit in using the road at the same time as the other trucks. Here, $g:\M\rightarrow \mathbb{R}$ is a nondecreasing function and $m_r(x)=\sum_{\ell=1}^M \mathbf{1}_{\{x_\ell=r\}}$ is the number of trucks that are using the road at time interval $r\in\mathcal{R}$. The increased utility can be justified by the fact that whenever there are many trucks on the road at the same time interval, they can potentially collaborate to form platoons and thereby increase the fuel efficiency. It should be noted that this extra utility is a function of the average velocity of the flow since trucks cannot save a significant amount of fuel through platooning whenever traveling at low velocities~\cite{AlAlam10,Alam11-thesis}. The function $g:\M\rightarrow \mathbb{R}$ describes the dependency of the platooning incentive on the number of trucks that are using the road at that time interval. Again, the choice of this function does not change the mathematical results presented in this paper, but it can help us to capture the relationship between the fuel saving and the number of the trucks on the road. For instance, $g(m_{x_j}(x))=m_{x_j}(x)$ shows that the vehicles can even benefit from a low number of trucks but $g(m_{x_j}(x))=m_{x_j}(x) \mathbf{1}_{m_{x_j}(x)\geq \tau}$ describes the case where the trucks do not benefit until they reach a critical number $\tau\in\mathbb{N}$. For the simulations, we use the first mapping.

Notice that in the utilities $U_i$ in~\eqref{eqn:util:1} and $V_j$ in~\eqref{eqn:util:2}, we introduced  congestion taxes for cars and trucks. Later, they are used to ensure that the described game is a potential game. Such a game admits at least one pure strategy Nash equilibrium and we can use joint strategy fictitious play and average strategy fictitious play to learn that equilibrium. A viable direction for future research could be to design taxing policies so as to enforce a socially optimal behavior, such as an optimal carbon emission profile, using mechanism design theory~\cite{RefWorks:83}.

\subsection{Congestion Game}

Now, we are ready to define a congestion game with two types of players using normal-form representation of strategic games~\cite{gibbons1992game,osborne1994course}.

\begin{definition}\textsc{(Car--Truck Congestion Game)}: \label{def:cartruck} A car--truck congestion game is defined as a tuple $\mathcal{G}=((\mathcal{R})_{i=1}^{N+M};((U_i)_{i=1}^N,(V_j)_{j=1}^M))$, that is, a combination of $N+M$ players with action space $(\mathcal{R})_{i=1}^{N+M}$ and utilities $((U_i)_{i=1}^N,(V_j)_{j=1}^M))$.
\end{definition}

A pure strategy Nash equilibrium for a car--truck congestion game is a pair $(z,x)\in\mathcal{R}^N\times \mathcal{R}^M$ such that 
\begin{align*}
U_i(z_i,z_{-i},x)&\geq U_i(z'_i,z_{-i},x), && \forall z'_i\in\mathcal{R}, && i\in\N, \\
V_j(x_j,x_{-j},z)&\geq V_j(x'_j,x_{-j},z), && \forall x'_j\in\mathcal{R}, && j\in\M.
\end{align*}
To prove the existence of a pure strategy Nash equilibrium or to use various learning algorithms for finding an equilibrium, we focus on a subclass of games, namely, potential games~\cite{monderer1996potential}. A car--truck congestion game is a potential game with potential function $\Phi:\mathcal{R}^N\times \mathcal{R}^M\rightarrow \mathbb{R}$ if 
\begin{align*}
\Phi(x,z_i,z_{-i})&-\Phi(x,z'_i,z_{-i})=U_i(z_i,z_{-i},x)-U_i(z'_i,z_{-i},x), && \forall i\in\N, \\
\Phi(x_j,x_{-j},z)&-\Phi(x'_j,x_{-j},z)=V_j(x_j,x_{-j},z)-V_j(x'_j,x_{-j},z), && \forall j\in\M.
\end{align*}
With these definitions in hand, we are ready to present the results of the paper.

\section{Existence of Potential Function} \label{sec:potential}
Atomic congestion games with one type of agents (corresponding to the case where $M=0$ or $N=0$) are known to admit a potential function even without congestion taxes~\cite{gameroad,roughgarden2007routing,monderer1996potential}. In this section, we show that this property does not hold for car--truck congestion games unless we devise an appropriate taxing scheme. 

\subsection{Necessary Condition for the Existence of a Potential Function}
Let $\Phi:\mathbb{R}^{N}\times \mathbb{R}^{M}\rightarrow \mathbb{R}$ be a given mapping. Define
\begin{equation*}
\begin{split}
\Delta_{x_j\rightarrow x'_j} \Phi(x,z)&=\Phi(x,z)-\Phi(x',z), \\
\Delta_{z_i\rightarrow z'_i} \Phi(x,z)&=\Phi(x,z)-\Phi(x,z'),
\end{split}
\end{equation*}
where $x'=(x'_j,x_{-j})$ and $z'=(z'_i,z_{-i})$. Using simple algebra, we can show that the operators commute, i.e.,
$$
\Delta_{z_i\rightarrow z'_i}\Delta_{x_j\rightarrow x'_j} \Phi(x,z)=\Delta_{x_j\rightarrow x'_j} \Delta_{z_i\rightarrow z'_i}\Phi(x,z).
$$
Now, we are ready to prove the following useful result.

\begin{proposition} \label{lemma:necessity} A car--truck congestion game admits a potential function only if
\begin{equation*}
\begin{split}
\Delta_{x_i\rightarrow x'_j}\Delta_{z_i\rightarrow z'_i} V_j(z,x)=\Delta_{z_i\rightarrow z'_i} \Delta_{x_i\rightarrow x'_j} U_i(z,x), 
\end{split}
\end{equation*}
for all $i\in\N$ and $j\in\M$.
\end{proposition} 

\begin{proof} Let $\Phi(x,z)$ be a potential function for the congestion game. Then, it must satisfy
\begin{equation} \label{eqn:necessity:proof:1}
\begin{split}
\Delta_{x_j\rightarrow x'_j} V_j(x,z)=\Delta_{x_j\rightarrow x'_j} \Phi(x,z).
\end{split}
\end{equation}
Let $x'=(x'_j,x_{-j})$ and $z'=(z'_i,z_{-i})$. Again, when noting that $\Phi(x,z)$ is a potential function, we get
\begin{subequations} \label{eqn:necessity:proof:2}
\begin{align}
\Phi(x,z)&=\Phi(x,z')+\Delta_{z_i\rightarrow z'_i} U_i(z,x)\\
\Phi(x',z)&=\Phi(x',z')+\Delta_{z_i\rightarrow z'_i} U_i(z,x')
\end{align}
\end{subequations}
Substituting~\eqref{eqn:necessity:proof:2} into~\eqref{eqn:necessity:proof:1} results in
\begin{equation*}
\begin{split}
\Delta_{x_j\rightarrow x'_j} V_j(x,z)\hspace{-.04in}=&\Phi(x,z)-\Phi(x',z)
\\=&\Delta_{x_j\rightarrow x'_j} \Phi(x,z')+\Delta_{z_i\rightarrow z'_i} U_i(z,x)-\Delta_{z_i\rightarrow z'_i} U_i(z,x')\\
=&\Delta_{x_j\rightarrow x'_j} \Phi(x,z')\hspace{-.04in}+\hspace{-.04in}\Delta_{z_i\rightarrow z'_i} \Delta_{x_i\rightarrow x'_j} U_i(z,x)
\\=&\Delta_{x_j\rightarrow x'_j} V_j(x,z')\hspace{-.04in}+\hspace{-.04in}\Delta_{z_i\rightarrow z'_i} \Delta_{x_i\rightarrow x'_j} U_i(z,x),
\end{split}
\end{equation*}
where the last equality follows from the definition of the potential function. Therefore, we get the identity in the statement of the theorem.
\end{proof}

This shows that it might not be possible to find a potential functions for the congestion game with two types of players.

\begin{corollary} Let $p_i^{\mathrm{c}}(z,x)=0$ for $i\in\N$ and $p_j^{\mathrm{t}}(z,x)=0$ for $j\in\M$. A car--truck congestion game admits a potential function only if $\beta=0$ or $g$ is equal to zero everywhere.
\end{corollary}

\begin{proof} First, by simple algebraic manipulations, we prove the identity in
\begin{equation} \label{eqn:longequation:1}
\begin{split}
\Delta_{x_i\rightarrow x'_j}\Delta_{z_i\rightarrow z'_i} V_j(z,x)&=\Delta_{x_i\rightarrow x'_j}\Delta_{z_i\rightarrow z'_i} \big(\xi_j^{\mathrm{t}}(x_j,T_j^{\mathrm{t}}) +v_{x_j}(z,x) +\beta v_{x_j}(z,x)g(m_{x_j}(x))\big)
\\&=\Delta_{x_i\rightarrow x'_j}\Delta_{z_i\rightarrow z'_i} \big(v_{x_j}(z,x) +\beta v_{x_j}(z,x)g(m_{x_j}(x))\big)
\\&=\Delta_{x_i\rightarrow x'_j}\big(v_{x_j}(z,x)\hspace{-.03in}-\hspace{-.03in}v_{x_j}(z',x)\hspace{-.03in}+\hspace{-.03in}\beta v_{x_j}(z,x)g(m_{x_j}(x))\hspace{-.03in}-\hspace{-.03in}\beta v_{x_j}(z',x)g(m_{x_j}(x))\big)
\\&=\Delta_{x_i\rightarrow x'_j}\big( a[\mathbf{1}_{x_j=z_i}-\mathbf{1}_{x_j=z'_i}][1-\beta g(m_{x_j}(x))]\big)
\\&=a[\mathbf{1}_{x_j=z_i}-\mathbf{1}_{x_j=z'_i}][1-\beta g(m_{x_j}(x))]-a[\mathbf{1}_{x'_j=z_i}-\mathbf{1}_{x'_j=z'_i}][1-\beta g(m_{x'_j}(x'))]
\\&=a[\mathbf{1}_{x_j=z_i}+\mathbf{1}_{x'_j=z'_i}-\mathbf{1}_{x_j=z'_i}-\mathbf{1}_{x'_j=z_i}]
\\&\hspace{.2in}-a\beta[\mathbf{1}_{x_j=z_i}-\mathbf{1}_{x_j=z'_i}]g(m_{x_j}(x))+a\beta[\mathbf{1}_{x'_j=z_i}-\mathbf{1}_{x'_j=z'_i}]g(m_{x'_j}(x'))
\\&=a[\mathbf{1}_{x_j=z_i}+\mathbf{1}_{x'_j=z'_i}-\mathbf{1}_{x_j=z'_i}-\mathbf{1}_{x'_j=z_i}]
\\&\hspace{.2in}+a\beta[\mathbf{1}_{x_j=z'_i}\mathbf{1}_{x'_j=z_i}-\mathbf{1}_{x_j=z_i}\mathbf{1}_{x'_j=z'_i}][1-\mathbf{1}_{z_j=z'_i}][g(m_{x_j}(x))+g(m_{x'_j}(x'))]
\end{split}
\end{equation}
Similarly, we can show that
\begin{equation*}
\begin{split}
\Delta_{z_i\rightarrow z'_i} \Delta_{x_i\rightarrow x'_j} &U_i(z,x)=a[\mathbf{1}_{x_j=z_i}+\mathbf{1}_{x'_j=z'_i}-\mathbf{1}_{x_j=z'_i}-\mathbf{1}_{x'_j=z_i}].
\end{split}
\end{equation*}
Therefore, following Proposition~\ref{lemma:necessity}, the introduced congestion game admits a potential function only if
\begin{equation*}
\begin{split}
\beta[\mathbf{1}_{x_j=z'_i}\mathbf{1}_{x'_j=z_i}&-\mathbf{1}_{x_j=z_i}\mathbf{1}_{x'_j=z'_i}][1-\mathbf{1}_{z_j=z'_i}][g(m_{x_j}(x))+g(m_{x'_j}(x'))]=0
\end{split}
\end{equation*}
for all $x,z$ and $x'_j,z'_i$. This is indeed only possible if $\beta=0$ or if $g$ is equal to zero everywhere.
\end{proof}

Potential games have many desirable attributes. For instance, these games always admit at least one pure strategy Nash equilibrium. In addition, many learning algorithms, such as, joint strategy fictitious play, are known to extract a pure strategy Nash equilibrium for potential games. Given these important properties, a natural question that comes to mind is that whether it is possible to guarantee the existence of a potential function by imposing appropriate congestion taxes. We answer this question in the next subsection.

\subsection{Imposing Taxes to Guarantee the Existence of a Potential Function} \label{subsec:taxforpotential}
In this subsection, we propose a taxing and a subsidy policy that guarantee the existence of a potential function for the car--truck congestion game.

\begin{theorem} \label{tho:1} Let each car $i\in\N$ pay the congestion tax
\begin{equation} \label{eq:tax}
p_i^{\mathrm{c}}(z,x)=a\beta\sum_{\ell=1}^{m_{z_i}(x)}g(\ell),
\end{equation}
for using the road at time interval $z_i\in\mathcal{R}$. Then, the car--truck congestion game is a potential game with the potential function
\begin{align*}
\Phi(x,z)=&\sum_{i=1}^N \xi_i^{\mathrm{c}}(z_i,T_i^{\mathrm{c}})+\sum_{j=1}^M \xi_j^{\mathrm{t}}(x_j,T_j^{\mathrm{t}})+\sum_{r=1}^R\sum_{k=1}^{n_r(x,z)}(ak+b)\\&+\sum_{r=1}^R\beta(an_r(x,z)+b) \sum_{\ell=1}^{m_r(x)}g(\ell)- a\beta\sum_{r=1}^R\sum_{\ell=1}^{m_r(x)}\sum_{k=1}^{\ell-1}g(k).
\end{align*}
Furthermore, this game admits at least one pure strategy Nash equilibrium.
\end{theorem}

\begin{proof} The proof of this lemma follows the same line of reasoning as in the proof of Proposition 4.1 in~\cite{gameroad}. First, we need to define the following notations
\begin{align*}
\Phi_1(x,z)&=\sum_{i=1}^N \xi_i^{\mathrm{c}}(z_i,T_i^{\mathrm{c}})+\sum_{j=1}^M \xi_j^{\mathrm{t}}(x_j,T_j^{\mathrm{t}}), \\[-.5em]
\Phi_2(x,z)&=\sum_{r=1}^R\sum_{k=1}^{n_r(x,z)}(ak+b), \\[-.5em]
\Phi_3(x,z)&=\sum_{r=1}^R\beta(an_r(x,z)+b) \sum_{\ell=1}^{m_r(x)}g(\ell), \\[-.5em]
\Phi_4(x,z)&=- a\beta\sum_{r=1}^R\sum_{\ell=1}^{m_r(x)}\sum_{k=1}^{\ell-1}g(k).
\end{align*}
Let us start by analyzing the trucks. If $x_j=x'_j$, the result trivially holds. Consequently, we consider the case where $x_j\neq x'_j$, which results in
\begin{equation*}
\begin{split}
\Phi(x_j,x_{-j},z)-\Phi(x'_j,x_{-j},z)=\sum_{k=1}^4\Phi_k(x_j,x_{-j},z)-\Phi_k(x'_j,x_{-j},z).
\end{split}
\end{equation*}
We continue the proof by considering each term of this summation separately. For the first term, clearly, we have
\begin{equation*}
\begin{split}
\Phi_1(x_j,x_{-j},z)-\Phi_1(x'_j,x_{-j},z) =\xi_j^{\mathrm{t}}(x_j,T_j^{\mathrm{t}}) -\xi_j^{\mathrm{t}}(x'_j,T_j^{\mathrm{t}}).
\end{split}
\end{equation*}
Let us define $x'=(x'_j,x_{-j})$. For the second term, we have
\begin{equation*}
\begin{split}
\Phi_2(x_j,x_{-j},z)-\Phi_2(x'_j,x_{-j},z)&=\sum_{r=1}^R \sum_{k=1}^{n_r(x,z)}(ak+b) -\sum_{r=1}^R\sum_{k=1}^{n_r(x',z)}(ak+b)\\&= \sum_{k=1}^{n_{x_j}(x,z)}(ak+b) +\sum_{k=1}^{n_{x'_j}(x,z)}(ak+b)\\ &\hspace{+.2in}-\sum_{k=1}^{n_{x_j}(x',z)}(ak+b) -\sum_{k=1}^{n_{x'_j}(x',z)}(ak+b),
\end{split}
\end{equation*}
where the second equality holds because of the fact that $n_r(x,z)= n_r(x',z)$ for all $r\neq x_j,x'_j$. Note that
\begin{equation} \label{eqn:indentity:n}
\begin{split}
n_{x_j}(x',z)=n_{x_j}(x,z)\hspace{-.04in}-\hspace{-.04in}1, \; n_{x'_j}(x,z)=n_{x'_j}(x',z)\hspace{-.04in}-\hspace{-.04in}1,
\end{split}
\end{equation}
and as a result,
\begin{equation*}
\begin{split}
\Phi_2(x_j,x_{-j},z)&-\Phi_2(x'_j,x_{-j},z) =(an_{x_j}(z,x)+b)-(an_{x'_j}(z,x')+b).
\end{split}
\end{equation*}
For the third term, we get the identity in
\begin{equation} \label{eqn:proof:long:1}
\begin{split}
\Phi_3(x_j,x_{-j},z)-\Phi_3(x'_j,x_{-j},z)&=\sum_{r=1}^R\hspace{-.03in} \beta(an_r(x,z)+b)\hspace{-.05in}\sum_{\ell=1}^{m_r(x)}\hspace{-.05in}g(\ell) \hspace{-.03in}-\hspace{-.03in}\sum_{r=1}^R\hspace{-.03in}\beta(an_r(x',z)+b)\hspace{-.05in}\sum_{\ell=1}^{m_r(x')}\hspace{-.05in}g(\ell)
\\&=\beta(an_{x_j}(x,z)+b)\sum_{\ell=1}^{m_{x_j}(x)}g(\ell)
+\beta(an_{x'_j}(x,z)+b)\sum_{\ell=1}^{m_{x'_j}(x)}g(\ell)
\\&\hspace{.2in}-\beta(an_{x_j}(x',z)+b)\hspace{-.05in}\sum_{\ell=1}^{m_{x_j}(x')} \hspace{-.05in}g(\ell) \hspace{-.03in}-\hspace{-.03in}\beta(an_{x'_j}(x',z)+b)\hspace{-.05in}\sum_{\ell=1}^{m_{x'_j}(x')}\hspace{-.05in}g(\ell)
\\&=\beta(an_{x_j}(x,z)+b)g(m_{x_j}(x)) -\beta(an_{x'_j}(x',z)+b)g(m_{x'_j}(x')) \\&\hspace{.2in}+a\beta\sum_{\ell=1}^{m_{x_j}(x)-1}g(\ell)-a\beta \sum_{\ell=1}^{m_{x'_j}(x')-1}g(\ell),
\end{split}
\end{equation}
where the last equality follows from using~\eqref{eqn:indentity:n} and the fact that $m_{x_j}(x')=m_{x_j}(x)-1 $ and $ m_{x'_j}(x)=m_{x'_j}(x')-1.$ Finally, using the same argument as in the case of the second term and the third term, we get
\begin{equation*}
\begin{split}
\Phi_4(x_j,x_{-j},z)&-\Phi_4(x'_j,x_{-j},z)=-a\beta\sum_{\ell=1}^{m_{x_j}(x)-1}g(\ell)
+a\beta\sum_{\ell=1}^{m_{x'_j}(x')-1}g(\ell).
\end{split}
\end{equation*}
Combining all these differences, we get
\begin{equation*}
\begin{split}
\Phi(x_j,x_{-j},z)-\Phi(x'_j,x_{-j},z)
=&\beta(an_{x_j}(x,z)+b)g(m_{x_j}(x)) -\beta(an_{x'_j}(x',z)\hspace{-.03in}+\hspace{-.03in}b)g(m_{x'_j}(x'))
\\&+\xi_j^{\mathrm{t}}(x_j,T_j^{\mathrm{t}}) -\xi_j^{\mathrm{t}}(x'_j,T_j^{\mathrm{t}})+(an_{x_j}(z,x)+b) -(an_{x'_j}(z,x')+b)
\\=&V_j(x_j,x_{-j},z)\hspace{-.03in}-\hspace{-.03in}V_j(x'_j,x_{-j},z).
\end{split}
\end{equation*}
Now, let us prove this fact for the cars as well. If $z_i=z'_i$, the result trivially holds. Thus, we investigate the case where $z_i\neq z'_i$. Similarly, we consider each term of the summation separately. For the first term, we have
\begin{equation*}
\begin{split}
\Phi_1(x,z_{i},z_{-i})&-\Phi_1(x,z'_{i},z_{-i})= \xi_i^{\mathrm{c}}(z_i,T_i^{\mathrm{c}})- \xi_i^{\mathrm{c}}(z'_i,T_i^{\mathrm{c}}).
\end{split}
\end{equation*}
We define the notation $z'=(z'_i,z_{-i})$. Following a similar reasoning as in the case of the trucks, for the second and the third terms, we get
\begin{equation*}
\begin{split}
\Phi_2(x,z_{i},z_{-i})&-\Phi_2(x,z'_{i},z_{-i}) =(an_{z_i}(z,x)+b)-(an_{z'_i}(z',x)+b),
\end{split}
\end{equation*}
and
\begin{equation*}
\begin{split}
\Phi_3(x,z_{i},z_{-i})\hspace{-.03in}-\hspace{-.03in}\Phi_3(x,z'_{i},z_{-i})&= a\beta\hspace{-.05in}\sum_{\ell=1}^{m_{z_i}(x)}\hspace{-.05in}g(\ell)-a\beta \hspace{-.05in}\sum_{\ell=1}^{m_{z'_i}(x)}\hspace{-.05in}g(\ell).
\end{split}
\end{equation*}
For the forth term, we get $\Phi_4(x,z_{i},z_{-i})-\Phi_4(x,z'_{i},z_{-i})=0$ 
since this term is only a function of $x$ which is not changed. Again, combining all these differences, we get
\begin{equation*}
\begin{split}
\Phi(x,z_{i},z_{-i})-\Phi(x,z'_{i},z_{-i})
=&\,(an_{z_i}(z,x)+b)-(an_{z'_i}(z',x)+b)
+\xi_i^{\mathrm{c}}(z_i,T_i^{\mathrm{c}})- \xi_i^{\mathrm{c}}(z'_i,T_i^{\mathrm{c}})\\& + a\beta\hspace{-.05in}\sum_{\ell=1}^{m_{z_i}(x)}\hspace{-.05in}g(\ell)- a\beta\hspace{-.05in}\sum_{\ell=1}^{m_{z'_i}(x)}\hspace{-.05in}g(\ell)
\\=&\,U_i(z_{i},z_{-i},x)-U_i(z'_{i},z_{-i},x).
\end{split}
\end{equation*}
Finally, note that every potential game admits at least one pure strategy Nash equilibrium~\cite{monderer1996potential}.
\end{proof}

\begin{remark} Note the tax $p_i^{\mathrm{c}}(z,x)$ grows quadratically with the number of the trucks that are using the road at that time interval if the mapping $g:\M\rightarrow \mathbb{R}$ is an affine function. Therefore, the congestion tax policy $p_i^{\mathrm{c}}(z,x)$ in Theorem~\ref{tho:1} forces the cars to avoid the time intervals that the trucks use to travel together.
\end{remark}

Instead of taxing the cars, we can also introduce a platooning subsidy for the trucks to get a potential game.

\begin{theorem} \label{lem:2} Let each truck $j\in\M$ receive the subsidy
\begin{equation} \label{eq:subsidy}
p_j^{\mathrm{t}}(x,z)=\beta (v_0-(an_{x_j}(z,x)+b))m_{x_j}(x),
\end{equation}
for a given $v_0\in\mathbb{R}$. Then, the car--truck congestion game is a potential game with the potential function
\begin{equation*}
\begin{split}
\Psi(x,z)=\,&\sum_{i=1}^N\xi_i^{\mathrm{c}}(z_i,T_i^{\mathrm{c}}) +\sum_{j=1}^M \xi_j^{\mathrm{t}}(x_j,T_j^{\mathrm{t}})+\sum_{r=1}^R \sum_{k=1}^{n_r(x,z)}(ak+b)+\beta v_0 \sum_{r=1}^R\sum_{\ell=1}^{m_r(x)}g(\ell).
\end{split}
\end{equation*}
Furthermore, this game admits at least one pure strategy Nash equilibrium.
\end{theorem}

\begin{proof} Let us start with trucks. Note that with the introduced policy, the utility of truck $j$ is equal
\begin{equation*}
\begin{split}
V_j(x_j,x_{-j},z)\hspace{-.03in}=\hspace{-.03in}\xi_j^{\mathrm{t}}(x_j,T_j^{\mathrm{t}})\hspace{-.03in}+\hspace{-.03in}v_{x_j}(z,x)\hspace{-.03in}+\hspace{-.03in}\beta v_0g(m_{x_j}(x)).
\end{split}
\end{equation*}
Let us define $x'=(x'_j,x_{-j})$. If $x_j=x'_j$, the result trivially holds. Therefore, without loss of generality, we consider the case where $x_j\neq x'_j$. In what follows, we examine each term in the cost function separately. First, we define $\Psi_1(x,z)=\sum_{i=1}^N\xi_i^{\mathrm{c}}(z_i,T_i^{\mathrm{c}}) +\sum_{j=1}^M \xi_j^{\mathrm{t}}(x_j,T_j^{\mathrm{t}})$. Now, it is easy to see that
$$
\Psi_1(x,z)-\Psi_1(x',z)=\xi_j^{\mathrm{t}}(x_j,T_j^{\mathrm{t}}) -\xi_j^{\mathrm{t}}(x'_j,T_j^{\mathrm{t}}).
$$
Second, we define $\Psi_2(x,z)=\sum_{r=1}^R \sum_{k=1}^{n_r(x,z)} (ak+b)$. For this term, we can show that
\begin{equation*}
\begin{split}
\Psi_2(x,z)\hspace{-.03in}-\hspace{-.03in}\Psi_2(x',z)&\hspace{-.03in}=\hspace{-.06in}\sum_{r=1}^R \hspace{-.03in}\sum_{k=1}^{n_r(x,z)}\hspace{-.03in}(ak+b) \hspace{-.03in}-\hspace{-.03in}\sum_{r=1}^R\hspace{-.03in}\sum_{k=1}^{n_r(x',z)}\hspace{-.03in}(ak+b)\\&\hspace{-.03in}=\hspace{-.03in} \sum_{k=1}^{n_{x_j}(x,z)}\hspace{-.03in}(ak+b) \hspace{-.03in}+\hspace{-.03in}\sum_{k=1}^{n_{x'_j}(x,z)}\hspace{-.03in}(ak+b)\hspace{-.03in}-\hspace{-.06in}\sum_{k=1}^{n_{x_j}(x',z)}\hspace{-.03in}(ak+b)\hspace{-.03in} -\hspace{-.03in}\sum_{k=1}^{n_{x'_j}(x',z)}\hspace{-.03in}(ak+b),
\end{split}
\end{equation*}
where the second equality holds because of the fact that $n_r(x,z)= n_r(x',z)$ for all $r\neq x_j,x'_j$. Noticing that $n_{x_j}(x',z)=n_{x_j}(x,z)-1$ and $n_{x'_j}(x,z)=n_{x'_j}(x',z)-1$, we know that
$$
\Psi_2(x,z)-\Psi_2(x',z)=(an_{x_j}(z,x)+b)-(an_{x'_j}(z,x')+b).
$$
Finally, we define $\Psi_3(x,z)=\sum_{r=1}^R\sum_{\ell=1}^{m_r(x)}g(\ell)$. In this case, we can show that
\begin{equation*}
\begin{split}
\Psi_3(x,z)-\Psi_3(x',z)=&\sum_{r=1}^R\sum_{\ell=1}^{m_r(x)}g(\ell) -\sum_{r=1}^R\sum_{\ell=1}^{m_r(x')}g(\ell)
\\=&\sum_{\ell=1}^{m_{x_j}(x)}g(\ell)+\sum_{\ell=1}^{m_{x'_j}(x)}g(\ell)-\sum_{\ell=1}^{m_{x_j}(x')}g(\ell)-\sum_{\ell=1}^{m_{x'_j}(x')}g(\ell)
\\=&g(m_{x_j}(x))-g(m_{x'_j}(x')).
\end{split}
\end{equation*}
Therefore, we get
\begin{equation*}
\begin{split}
\Psi(x,z)-\Psi(x',z)=&\Psi_1(x,z)-\Psi_1(x',z)+\Psi_2(x,z)-\Psi_2(x',z)+\beta v_0 (\Psi_3(x,z)-\Psi_3(x',z))\\=&\xi_j^{\mathrm{t}}(x_j,T_j^{\mathrm{t}})-\xi_j^{\mathrm{t}}(x'_j,T_j^{\mathrm{t}})+v_{x_j}(x,z)-v_{x'_j}(x',z)+\beta v_0(g(m_{x_j}(x))-g(m_{x'_j}(x')))\\=&V_j(x_j,x_{-j},z)-V_j(x'_j,x_{-j},z).
\end{split}
\end{equation*}
The proof for cars follows the same line of reasoning. 
\end{proof}

\begin{remark} Note that if $v_0$ is greater than the average velocity of the flow, the trucks get paid to use the road at the same time as their peers. This way the government incentivizes the trucks to form platoons. This subsidy is technically the difference between the amount of the fuel that the trucks would have saved if they formed a platoon at velocity $v_0$ instead of the actual average velocity of the traffic flow $an_r(z,x)+b$. Therefore, the trucks would benefit from traveling together even at low velocities (which is a scenario where the trucks do not increase their fuel efficiency significantly through platooning). However, if $v_0$ is smaller than the average velocity of the flow, we reduce the extra utility that the trucks would receive from traveling together (and technically $p_j^{\mathrm{t}}(x,z)$ becomes a tax rather than a subsidy). Therefore, it becomes less likely for the trucks to stick together. To emphasize the fact that we are willing to pay the trucks rather than taxing them (and hence, dealing with the first scenario), we call $p_j^{\mathrm{t}}(x,z)$ a subsidy.
\end{remark}

\section{Joint Strategy Fictitious Play} \label{sec:JSFP}
We start by briefly introducing the learning algorithm and, then, analyzing its convergence.

\subsection{Learning Algorithm}

\begin{algorithm}[t]
\renewcommand{\baselinestretch}{1}
\caption{\label{alg:1} Joint strategy fictitious play for learning a Nash equilibrium. }
\begin{algorithmic}[1]
\begin{small}
\REQUIRE $p\in(0,1)$
\ENSURE $(x^*,z^*)$
\FOR{$t=0,1,\dots$}
\FOR{$i=1,\dots,N$}
\STATE Calculate $z'_i\in\argmax_{r\in\mathcal{R}}\hat{U}_i(r;t-1)$
\IF{$U_i(z'_i,z_{-i}(t-1),x(t-1))\leq U_i(z_i(t-1),z_{-i}(t-1),x(t-1))$}
\STATE $z_i(t)\leftarrow z_i(t-1)$
\ELSE
\STATE With probability $1-p$, $z_i(t)\leftarrow z_i(t-1)$, otherwise $z_i(t)\leftarrow z'_i$
\ENDIF
\FOR{$j=1,\dots,M$}
\STATE Calculate $x'_j\in\argmax_{r\in\mathcal{R}}\hat{V}_j(r;t-1)$
\IF{ $V_j(z(t-1),x'_j,x_{-j}(t-1))\leq V_j(z(t-1),x_j(t-1),x_{-j}(t-1))$}
\STATE $x_j(t)\leftarrow x_j(t-1)$
\ELSE
\STATE With probability $1-p$, $x_j(t)\leftarrow x_j(t-1)$, otherwise $x_j(t)\leftarrow x'_j$
\ENDIF
\ENDFOR
\ENDFOR
\ENDFOR
\end{small}
\end{algorithmic}
\end{algorithm}

Assume that the agents follow the joint strategy fictitious play algorithm~\citep{marden2009joint}. To do so, the agents calculate an average utility given the history of the actions. At time step $t\in\mathbb{N}_{0}$, car~$i\in\N$ computes $\hat{U}_i(r;t)$ using the recursive equation
\begin{equation} \label{eqn:1}
\hat{U}_i(r;t)=(1-\lambda_t)\hat{U}_i(r;t-1) +\lambda_t U_i(r,z_{-i}(t),x(t)), 
\end{equation}
with the initial condition $\hat{U}_i(r;-1)=\xi_i^{\mathrm{c}}(r,T_i^{\mathrm{c}})$ for all $r\in\mathcal{R}$. In~\eqref{eqn:1}, $\lambda_t\in(0,1]$ is a forgetting factor which captures the extent that the agents forget the actions from the past. If $\lambda_t=1$, the agents are myopic (i.e., only consider the actions from the previous time step) while if $\lambda_t=1/t$, the agents value the whole history at the same level. Following the same approach, truck $j\in\M$ calculates $\hat{V}_j(r;t)$ using the recursive equation
$$
\hat{V}_j(r;t)=(1-\lambda_t)\hat{V}_j(r;t-1) +\lambda_tV_j(r,x_{-j}(t),z(t)),
$$
with $\hat{V}_j(r;-1)=\xi_j^{\mathrm{t}}(r,T_j^{\mathrm{t}})$ for all $r\in\mathcal{R}$. Algorithm~\ref{alg:1} shows the joint strategy fictitious play for the car--truck congestion game. 

\subsection{Convergence Analysis}

Noting that with appropriate taxes the introduced congestion game is a potential game, we can use the result of \citep{marden2009joint} to conclude the convergence of the learning algorithm.

\begin{theorem} Let the action profile of the agents be generated by the joint strategy fictitious play in Algorithm~\ref{alg:1}. Assume that $\lambda_t=\lambda\in(0,1)$ or $\lambda_t=1/t$ for all $t\in\mathbb{N}$. Then, this action profile almost surely converges to a pure strategy Nash equilibrium of the car--truck congestion game, if either the cars pay the congestion tax $p_i^{\mathrm{c}}(z,x)$ in~\eqref{eq:tax} or the trucks receive the platooning subsidy $p_j^{\mathrm{t}}(x,z)$ in~\eqref{eq:subsidy}.
\end{theorem}

\begin{proof} The proof is a consequence of combining Theorems~2.1 and~3.1 in~\citep{marden2009joint} with Theorems~\ref{tho:1} and~\ref{lem:2}.
\end{proof}

Note that the joint strategy fictitious play might be restrictive in some aspects. For instance, all the agents must have access to all the individual decisions taken by the other agents to calculate the average cost function. In the next section, we adapt the average strategy fictitious play introduced in~\citep{gameroad} as an alternative. This learning algorithm requires instead a central node to broadcast the congestion prediction (i.e., an average of all the players actions) for all time intervals per day.

\section{Average Strategy Fictitious Play} \label{sec:ASFP}
First, we introduce the average strategy fictitious play and study its convergence by extending parts of the proofs in~\citep{gameroad}.

\subsection{Learning Algorithm}
Before introducing the learning algorithm, we have to make the following standing assumptions: 
\begin{assumption} The congestion tax policies satisfy
\begin{itemize}
\item $p_i^{\mathrm{c}}(z,x)$, $i\in\N$, is only a function of $n_{z_i}(x,z),m_{z_i}(x)$;
\item $p_j^{\mathrm{t}}(x,z)$, $j\in\M$, is only a function of $n_{x_j}(x,z),m_{x_j}(x)$.
\end{itemize} 
\end{assumption}

This assumption means that the congestion tax can only be function of the traffic flow rather than the individual actions of the agents. The congestion taxing policy that we introduced in the previous section satisfies this assumption. To emphasize this fact, from now on, we write $p_i^{\mathrm{c}}(n_{z_i}(x,z),m_{z_i}(x))$ and $p_j^{\mathrm{t}}(n_{x_j}(x,z),m_{x_j}(x))$ with some abuse of notation.

Now, we can introduce the average strategy fictitious play. To initialize the algorithm, we let the agents pick an arbitrary action from the set $\mathcal{R}$ at the first time step. We assume that there exists a central node\footnote{This central node is assumed to be a not-for-profit organization. Therefore, it is not trying to optimize its income or loss (i.e., the summation of the received taxes or the distributed subsidies) and, hence, it would not strategically deviate from the intended algorithm. Certainly, introducing a mechanism with profitable organizations as a central node can be a viable avenue for future research (to attract the private sector for implementing this part). } that can observe the traffic flow at each time interval. This central node uses the following recursive update laws to calculate the average number of the cars and trucks in each time interval
\begin{align*}
\bar{n}^{\mathrm{c}}_r(t)&=(1-\lambda)\bar{n}^{\mathrm{c}}_r(t-1)+ \lambda \sum_{\ell=1}^N \mathbf{1}_{\{z_\ell(t)=r\}},
\\[-.5em]
\bar{n}^{\mathrm{t}}_r(t)&=(1-\lambda)\bar{n}^{\mathrm{t}}_r(t-1)+ \lambda \sum_{\ell=1}^M \mathbf{1}_{\{x_\ell(t)=r\}},
\end{align*}
with $\bar{n}^{\mathrm{c}}_r(0)=\sum_{\ell=1}^N \mathbf{1}_{\{z_\ell(0)=r\}}$ and $\bar{n}^{\mathrm{t}}_r(0)=\sum_{\ell=1}^M \mathbf{1}_{\{x_\ell(0)=r\}}$ for all $r\in\mathcal{R}$. The superscripts $\mathrm{c}$ and $\mathrm{t}$ show that the aforementioned property is related to the cars or the trucks, respectively. In these recursive update laws, we should choose the forgetting factor $\lambda\in(0,1)$ to capture the extent with which we value the congestion information from the past. We can think of the numbers $\bar{n}^{\mathrm{c}}_r(t)$ and $\bar{n}^{\mathrm{t}}_r(t)$ as the forecasts that the central node (e.g., the department of transportation, the radio station, etc) announces on a day-to-day basis about the traffic flow for each time interval of the day. These values have a memory to remember the congestion in earlier days and get updated based on the actual observation of the traffic flow every midnight.

Additionally, car $i\in\N$ and truck $j\in\M$ keep track of the average number of times that they have chosen  $r\in\mathcal{R}$ following the recursive update laws
\begin{align*}
\bar{w}^{\mathrm{c}}_{r,i}(t)&=(1-\lambda)\bar{w}^{\mathrm{c}}_{r,i} (t-1)+ \lambda\mathbf{1}_{\{z_i(t)=r\}},
\\[-.5em]
\bar{w}^{\mathrm{t}}_{r,j}(t)&=(1-\lambda)\bar{w}^{\mathrm{t}}_{r,j} (t-1)+ \lambda\mathbf{1}_{\{x_j(t)=r\}},
\end{align*}
with $\bar{w}^{\mathrm{c}}_{r,i}(0)=\mathbf{1}_{\{z_i(0)=r\}}$ and $\bar{w}^{\mathrm{t}}_{r,j}(0)=\mathbf{1}_{\{x_j(0)=r\}}$ for all $r\in\mathcal{R}$. Finally, for all $i\in\N$ and $j\in\M$, we define the new ``average'' cost functions in
\begin{subequations} \label{eqn:proof:long:2-3}
\begin{align} 
\tilde{V}_j(r;t)=&[a(\bar{n}^{\mathrm{c}}_{r}(t) +\bar{n}^{\mathrm{t}}_{r}(t)-\bar{w}^{\mathrm{t}}_{r,j}(t)+1)+b]  \nonumber \\ \nonumber &+\beta[a(\bar{n}^{\mathrm{c}}_{r}(t) +\bar{n}^{\mathrm{t}}_{r}(t)-\bar{w}^{\mathrm{t}}_{r,j}(t)+1) +b] g(\bar{n}^{\mathrm{t}}_{r}(t) -\bar{w}^{\mathrm{t}}_{r,j}(t)+1) \\&+\xi_j^{\mathrm{t}}(r,T_j^{\mathrm{t}}) + p_j^{\mathrm{t}}(\bar{n}^{\mathrm{c}}_{r}(t) +\bar{n}^{\mathrm{t}}_{r}(t)-\bar{w}^{\mathrm{t}}_{r,j}(t) +1,\bar{n}^{\mathrm{t}}_{r}(t) -\bar{w}^{\mathrm{t}}_{r,j}(t)+1), \label{eqn:proof:long:3} \\
\tilde{U}_i(r;t)=&\xi_i^{\mathrm{c}}(r,T_i^{\mathrm{c}}) +[a(\bar{n}^{\mathrm{c}}_{r}(t) +\bar{n}^{\mathrm{t}}_{r}(t) -\bar{w}^{\mathrm{c}}_{r,i}(t)+1)+b] \nonumber
\\&+p_i^\mathrm{c}(\bar{n}^{\mathrm{c}}_{r}(t) +\bar{n}^{\mathrm{t}}_{r}(t)-\bar{w}^{\mathrm{c}}_{r,i}(t)+1, \bar{n}^{\mathrm{t}}_{r}(t)). \label{eqn:proof:long:2}
\end{align}
\end{subequations}
Now, if we follow Algorithm~\ref{alg:2}, we expect to converge to a Nash equilibrium.

\begin{algorithm}[t]
\renewcommand{\baselinestretch}{1}
\begin{small}
\caption{\label{alg:2} Average strategy fictitious play for learning a Nash equilibrium. }
\begin{algorithmic}[1]
\REQUIRE $p\in(0,1)$
\ENSURE $(x^*,z^*)$
\FOR{$t=1,2,\dots$}
\FOR{$i=1,\dots,N$}
\STATE Calculate $z'_i\in\argmax_{r\in\mathcal{R}}\tilde{U}_i(r;t-1)$
\IF{ $U_i(z'_i,z_{-i}(t-1),x(t-1))\leq U_i(z_i(t-1),z_{-i}(t-1),x(t-1))$}
\STATE $z_i(t)\leftarrow z_i(t-1)$
\ELSE
\STATE With probability $1-p$, $z_i(t)\leftarrow z_i(t-1)$, otherwise $z_i(t)\leftarrow z'_i$
\ENDIF
\FOR{$j=1,\dots,M$}
\STATE Calculate $x'_j\in\argmax_{r\in\mathcal{R}}\tilde{V}_j(r;t-1)$
\IF{ $V_j(z(t-1),x'_j,x_{-j}(t-1))\leq V_j(z(t-1),x_j(t-1),x_{-j}(t-1))$}
\STATE $x_j(t)\leftarrow x_j(t-1)$
\ELSE
\STATE With probability $1-p$, $x_j(t)\leftarrow x_j(t-1)$, otherwise $x_j(t)\leftarrow x'_j$
\ENDIF
\ENDFOR
\ENDFOR 
\ENDFOR
\end{algorithmic}
\end{small}
\end{algorithm}

\subsection{Convergence Analysis}
First, we need to prove an intermediate lemma which shows that if Algorithm~\ref{alg:2} reaches a Nash equilibrium, it stays there forever.

\begin{lemma} \label{lem:3} Let each truck $j\in\M$ receive the subsidy
$$
p_j^{\mathrm{t}}(x,z)=\beta (v_0-(an_{x_j}(z,x)+b))m_{x_j}(x),
$$
for a given $v_0\in\mathbb{R}$. If $x(t)$ and $z(t)$, generated by Algorithm~\ref{alg:2}, is a pure strategy Nash equilibrium, and  $z_i(t)\in\argmax_{r\in \mathcal{R}}\tilde{U}_i(r;t-1)$ for all $i\in\N$ and $x_j(t)\in\argmax_{r\in\mathcal{R}}\tilde{V}_j(r;t-1)$ for all $j\in\M$, then $x(t')=x(t)$ and $z(t')=z(t)$ for all $t'\geq t$.
\end{lemma}

\begin{proof} The proof of this lemma follows the same line of reasoning as in the proof of Proposition~4.2 in~\citep{gameroad}. Here, we only prove the results for the trucks as the proof for the cars is technically the same. First, note that for all $r\in\mathcal{R}$, we get
\begin{subequations} \label{eqn:update:ASFP:1-2}
\begin{equation} \label{eqn:update:ASFP:1}
\begin{split}
\bar{n}^{\mathrm{c}}_{r}(t) +\bar{n}^{\mathrm{t}}_{r}(t)-\bar{w}^{\mathrm{t}}_{r}(t)\hspace{-.03in}=\hspace{-.03in}&\;(1\hspace{-.03in}-\hspace{-.03in}\lambda)\bar{n}^{\mathrm{c}}_r(t-1)\hspace{-.03in}+\hspace{-.03in}\lambda \sum_{\ell=1}^N \hspace{-.03in}\mathbf{1}_{\{z_\ell(t)=r\}}\hspace{-.03in}+\hspace{-.03in}(1\hspace{-.03in}-\hspace{-.03in}\lambda)\bar{n}^{\mathrm{t}}_r(t-1)\hspace{-.03in}+\hspace{-.03in}\lambda \sum_{\ell=1}^M \hspace{-.03in}\mathbf{1}_{\{x_\ell(t)=r\}} \\&-(1-\lambda)\bar{w}^{\mathrm{t}}_{r,j}(t-1)- \lambda\mathbf{1}_{\{x_j(t)=r\}}
\\\hspace{-.03in}=\hspace{-.03in}&\;(1\hspace{-.03in}-\hspace{-.03in}\lambda)(\bar{n}^{\mathrm{c}}_{r}(t-1)\hspace{-.03in}+\hspace{-.03in}\bar{n}^{\mathrm{t}}_{r}(t-1)\hspace{-.03in}-\hspace{-.03in}\bar{w}^{\mathrm{t}}_{r}(t-1))\hspace{-.03in}+\hspace{-.03in}\lambda (n_r(x(t),z(t))\hspace{-.03in}-\hspace{-.03in}\mathbf{1}_{\{x_j(t)=r\}}),
\end{split}
\end{equation}
\begin{equation} \label{eqn:update:ASFP:2}
\begin{split}
\bar{n}^{\mathrm{t}}_{r}(t) -\bar{w}^{\mathrm{t}}_{r,j}(t)=&\;(1-\lambda) \bar{n}^{\mathrm{t}}_r(t-1)+ \lambda\sum_{\ell=1}^M \mathbf{1}_{\{x_\ell(t)=r\}} -(1-\lambda)\bar{w}^{\mathrm{t}}_{r,j}(t-1)- \lambda\mathbf{1}_{\{x_j(t)=r\}}
\\=&\;(1-\lambda)(\bar{n}^{\mathrm{t}}_{r}(t-1) -\bar{w}^{\mathrm{t}}_{r,j}(t-1))+\lambda(m_r(x(t)) -\mathbf{1}_{\{x_j(t)=r\}}).
\end{split}
\end{equation}
\end{subequations}
Now, using these update laws and the proposed subsidy policy in~\eqref{eq:subsidy}, we get
\begin{equation*}
\begin{split}
\tilde{V}_j(r;t)=&\;\xi_j^{\mathrm{t}}(r,T_j^{\mathrm{t}}) +a(\bar{n}^{\mathrm{c}}_{r}(t) +\bar{n}^{\mathrm{t}}_{r}(t)- \bar{w}^{\mathrm{t}}_{r}(t)+1)+b+\beta v_0(\bar{n}^{\mathrm{t}}_{r}(t)-\bar{w}^{\mathrm{t}}_{r,j}(t)+1)\\= &\;\xi_j^{\mathrm{t}}(r,T_j^{\mathrm{t}}) + a(1-\lambda)(\bar{n}^{\mathrm{c}}_{r}(t-1) +\bar{n}^{\mathrm{t}}_{r}(t-1)- \bar{w}^{\mathrm{t}}_{r}(t-1))\hspace{-.03in}+\hspace{-.03in}a(\lambda (n_{r}(x(t),z(t))\hspace{-.03in}-\hspace{-.03in}\mathbf{1}_{\{x_j(t)=r\}})\hspace{-.03in}+\hspace{-.03in}1)\\&+b+\beta v_0(1-\lambda)(\bar{n}^{\mathrm{t}}_{r}(t-1) -\bar{w}^{\mathrm{t}}_{r,j}(t-1))+\beta v_0(\lambda(m_{r}(x(t))-\mathbf{1}_{\{x_j(t)=r\}})+1)
\\=&\;(1-\lambda)\tilde{V}_j(r;t-1)+\lambda V_j(r,x_{-j}(t),z(t)).
\end{split}
\end{equation*}
Therefore, we can prove that
\begin{equation*}
\begin{split}
\tilde{V}_j(x_j(t);\hspace{-.02in}t)&\hspace{-.04in}=\hspace{-.04in} (1\hspace{-.04in}-\hspace{-.04in}\lambda) \tilde{V}_j(x_j(t);t-1) \hspace{-.04in}+\hspace{-.04in}\lambda V_j(x_j(t),x_{\hspace{-.01in}-\hspace{-.01in}j}(t),\hspace{-.02in}z(t))
\\ &\hspace{-.04in}\geq\hspace{-.04in} (1\hspace{-.04in}-\hspace{-.04in}\lambda)\tilde{V}_j(r;t-1) \hspace{-.04in}+\hspace{-.04in}\lambda V_j(r,x_{-j}(t),z(t))
\\&\hspace{-.04in}=\hspace{-.04in}\tilde{V}_j(r;t)
\end{split}
\end{equation*}
for any $r\in\mathcal{R}$, where the inequality is direct consequence of the fact that the pair $x(t)$ and $z(t)$ is a pure strategy Nash equilibrium and $x_j(t)\in\argmax_{r\in\mathcal{R}}\tilde{V}_j(r;t-1)$ for all $j\in\M$. Thus, $x_j(t)\in\argmax_{r\in\mathcal{R}}\tilde{V}_j(r;t)$ and as a result, we get $x_j(t+1)=x_j(t)$ (following Algorithm~\ref{alg:2}). Now, using a simple mathematical induction, we can show $x_j(t+k)=x_j(t)$ for all $k\in\mathbb{N}$.
\end{proof}

\begin{theorem} Let the action profile of the agents be generated by the average strategy fictitious play in Algorithm~\ref{alg:2}. Then, this action profile almost surely converges to a pure strategy Nash equilibrium of the car--truck congestion game, if the trucks receive the platooning subsidy $p_j^{\mathrm{t}}(x,z)$ in~\eqref{eq:subsidy}.
\end{theorem}

\begin{proof} The proof follows from using Theorem~\ref{lem:2} and Lemma~\ref{lem:3} in the proof of Theorem~4.1 in~\citep{gameroad}.
\end{proof}

\section{Numerical Example} \label{sec:numericalexample}
Let us assume that $N=10000$ cars and $M=100$ trucks are using the segment of the highway illustrated in Figure~\ref{figuremap} from 7:00am to 9:00am on a daily basis. We divide the time horizon into eight equal non-overlapping intervals. Hence, we fix the action set as $\mathcal{R}=\{1,\dots,8\}$, where each number represents an interval of $15\,\mbox{min}$. Let $T_i^{\mathrm{c}}$, $i\in\N$, be randomly chosen from the set $\mathcal{R}$ using the discrete distribution
$$
\mathbb{P}\{T_i^{\mathrm{c}}=n\}=\left\{\begin{array}{ll}1/6, & n=2,4, \\ 1/4, & n=3, \\ 1/12, & \mbox{otherwise}. \end{array}\right.
$$
Let us also use a similar probability distribution to extract $T_j^{\mathrm{t}}$, $j\in\M$. Hence, we consider the case where the drivers statistically prefer to use the road at $r=3$ which corresponds to 7:30am to 7:45am. Let $\alpha_i^{\mathrm{c}}$, $i\in\N$, and $\alpha_j^{\mathrm{t}}$, $j\in\M$, be randomly generated following a uniform distribution within the interval $[-7.5,-2.5]$. Finally, let $a=-0.0110$ and $b=84.9696$ as discussed in Section~\ref{sec:problemsetup}.

\begin{figure}[!t]
\centering
\includegraphics[width=0.5\linewidth]{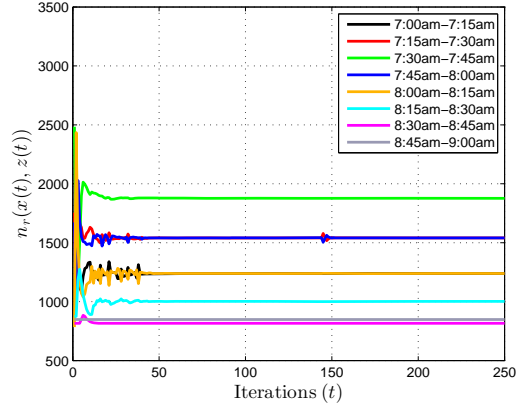}
\caption{\label{figure1} $n_{r}(x(t),z(t))$, $r\in\mathcal{R}$, versus the iteration number for $\beta=10^{-3}$ when using the joint strategy fictitious play in Algorithm~\ref{alg:1} with $p=0.4$ and $\lambda_t=3\times 10^{-2}$ for all $t\in\mathbb{N}_{0}$. }
\end{figure}

\begin{figure}[!t]
\centering
\includegraphics[width=0.5\linewidth]{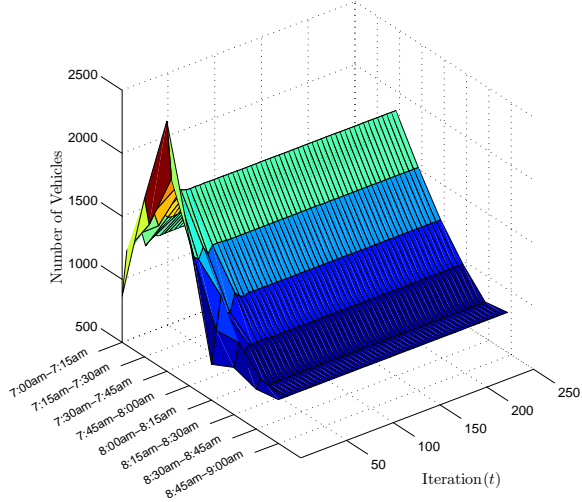}
\caption{\label{3dsurf} Number of the vehicles in each time interval for $\beta=10^{-3}$ when using the joint strategy fictitious play in Algorithm~\ref{alg:1} with $p=0.4$ and $\lambda_t=3\times 10^{-2}$ for all $t\in\mathbb{N}_{0}$. }
\end{figure}

\begin{figure}[!t]
\centering
\includegraphics[width=0.5\linewidth]{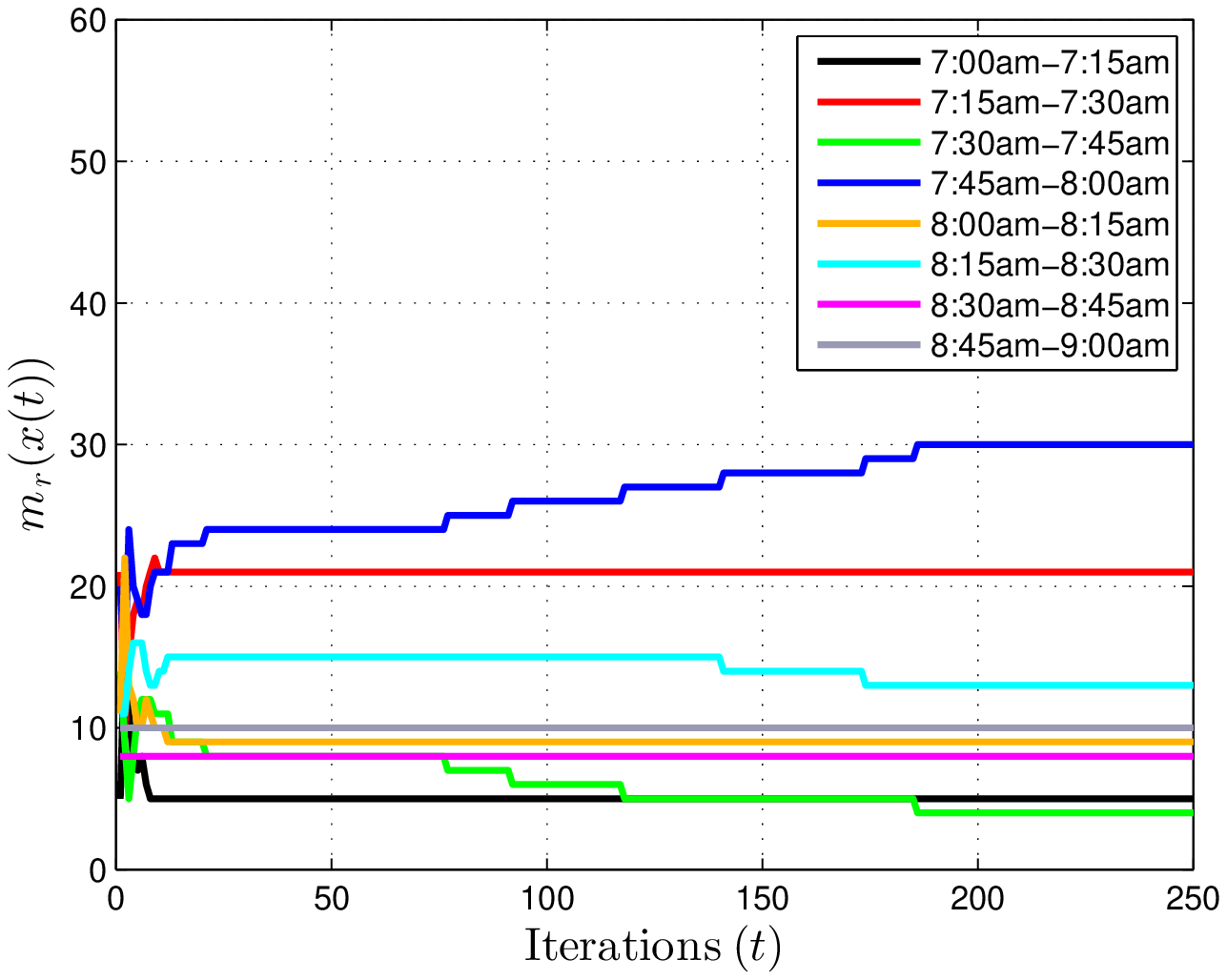}
\caption{\label{figure2} $m_{r}(x(t))$, $r\in\mathcal{R}$, versus the iteration number for $\beta=10^{-3}$ when using the joint strategy fictitious play in Algorithm~\ref{alg:1} with $p=0.4$ and $\lambda_t=3\times 10^{-2}$ for all $t\in\mathbb{N}_{0}$. }
\end{figure}

\begin{figure}[!t]
\centering
\includegraphics[width=0.5\linewidth]{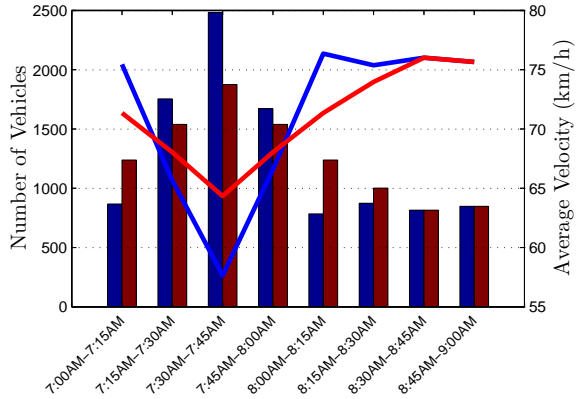}
\caption{\label{figure:averagevelocity} Number of the vehicles and the average velocity of the traffic flow in each time interval for the case where the drivers neglect the congestion in their decision making (blue) and for the learned pure strategy Nash equilibrium (red). }
\end{figure}

\subsection{Learning Algorithm Performance}
In this subsection, we start by simulating the joint strategy fictitious play in Algorithm~\ref{alg:1}. Let us fix $\beta=10^{-3}$, $p=0.4$, and $\lambda_t=3\times 10^{-2}$ for all $t\in\mathbb{N}_{0}$. Figure~\ref{figure1} illustrates the number of the vehicles (both cars and trucks) that are using a specific time interval to commute $n_{r}(x(t),z(t))$, $r\in\mathcal{R}$, as a function of the iteration number. As can be seen in this figure, the learning algorithm converges to a pure strategy Nash equilibrium in this example relatively fast\footnote{Recall that there are $|\mathcal{R}|^{M+N}$ possible action combinations in a car--truck congestion game. Therefore, in this example,  we have $8^{10100}\simeq 10^{9100}$ possible action combinations. To put this number into perspective, recall that there are around $10^{80}$ atoms in the visible universe. }. Figure~\ref{3dsurf} shows the evolution of the traffic distribution. Figure~\ref{figure2} shows the number of trucks $m_{r}(x(t))$, $r\in\mathcal{R}$, that are using the road on various time intervals. For instance, at the learned Nash equilibrium, thirty trucks use the time interval 7:45am to 8:00am while at the same time, most of them avoid using 7:15am to 7:30am because it is highly congested (and they would not save much fuel if they commute at this time).

\subsection{Nash Equilibrium Efficiency}
Figure~\ref{figure:averagevelocity} shows the number of the vehicles in each time interval and the corresponding average velocity in that time interval. The blue color denotes the case where the drivers do not consider the congestion in their decision making; i.e., they commute whenever pleases them, $z_i=T_i^\mathrm{c}$ for all $i\in\N$ and $x_j=T_j^\mathrm{t}$ for all $j\in\M$. The red color denotes the case where the drivers implement the pure strategy Nash equilibrium that they have learned using Algorithm~\ref{alg:1}. As we can see in this figure, the proposed congestion game reduces the average commuting time (increases the average velocity). Following~\cite{Vockingnisan2007algorithmic}, we can define the social cost
\begin{equation*}
\begin{split}
S(x,z)&=\min_{r\in\mathcal{R}} v_r(z,x)\\&=\min_{r\in\mathcal{R}} an_r(x,z)+b
\\&= a(\max_{r\in\mathcal{R}}n_r(x,z))+b,
\end{split}
\end{equation*}
where the last equality holds because of the fact that $a<0$. This social cost is the worst-case average velocity of the traffic flow\footnote{This cost function is an example of a Rawlsian social cost function (i.e., the worst-case cost function of the players). Another possible choice of social cost function is a utilitarian social cost function (i.e., summation of the individual cost functions of all the players); see~\cite[p.\,413]{2006general} for more information regarding the difference between these two categories of social cost functions.}. Another definition of social cost could be the total fuel consumption or the overall carbon emission. In a utopia, the government should be able to implement a global solution of the optimization problem
$$
(x^\bullet,z^\bullet)\in\argmax_{(z,x)\in\mathcal{R}^N\times \mathcal{R}^M}S(x,z),
$$
to achieve the lowest congestion at all time intervals. However, this solution cannot be implemented in a society with strategic (selfish) agents since they have no incentive for following a socially optimal decision $(x^\bullet,z^\bullet)$. Note that since $a<0$, we have
\begin{equation*}
\begin{split}
(x^\bullet,z^\bullet)&\in\argmax_{(z,x)\in\mathcal{R}^N\times \mathcal{R}^M} \min_{r\in\mathcal{R}} \; an_r(x,z)+b\\&\in\argmin_{(z,x)\in\mathcal{R}^N\times \mathcal{R}^M}\max_{r\in\mathcal{R}} \; n_r(x,z),
\end{split}
\end{equation*}
and as a result, we get
\begin{align*}
S(x^\bullet,z^\bullet)&=a\left\lceil \frac{N+M}{|\mathcal{R}|} \right\rceil+b
\\&=71.0766\,\mathrm{km/h}.
\end{align*}
Therefore, we have
$$
\frac{S(x^\bullet,z^\bullet)}{S(x^*,z^*)}=1.1048,
$$
which shows that the acquired pure strategy Nash equilibrium $(x^*,z^*)$ is not efficient with respect to the introduced welfare function\footnote{It is worth mentioning that if we choose the potential function $\Phi$ in Theorem~\ref{tho:1} as the social welfare function, the learned Nash equilibrium is indeed efficient since Algorith,~\ref{alg:1} results in  a local maximizer of this potential function. However, such a choice does not have any practical implications. }. However, it is somewhat better than the case where the drivers do not consider the congestion in their decision making (i.e. they travel whenever pleases them) as
$$
\frac{S(x^\bullet,z^\bullet)}{S(\{T_j^\mathrm{t}\}_{j=1}^M, \{T_i^\mathrm{c}\}_{i=1}^N)}=1.2330.
$$

\subsection{Robustness of the Learning Algorithm}
Let us now consider the case where on the fiftieth day of learning (i.e., iteration $t=50$) an unexpected behavior (e.g., a traffic accident) significantly decreases the average velocity of the traffic flow during 7:15am and 8:00am (i.e., for $r=2,3,4$). To reflect this matter in the simulations, we assume that on the fiftieth iteration, the average velocity for $r=2,3,4$ is given by $(an_r(x(t),z(t))+b)/10$. Figure~\ref{figure:Accidnet} illustrates the number of vehicles that are using a specific time interval to commute $n_{r}(x(t),z(t))$, $r\in\mathcal{R}$, as a function of the iteration numbers. Note that there is a sudden drop in the number of the vehicles that are using the time intervals corresponding to $r=2,3,4$ for a while (around twenty iterations) after the accident. However, the learning process recovers the Nash equilibrium after another fifty iterations.

\begin{figure}[!t]
\centering
\includegraphics[width=0.5\linewidth]{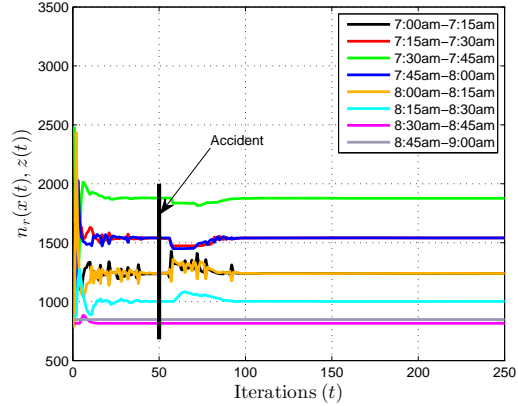}
\caption{\label{figure:Accidnet} $n_{r}(x(t),z(t))$, $r\in\mathcal{R}$, versus the iteration number when an unexpected behavior (e.g., an accident) disrupt the traffic flow on the fiftieth day of learning.  }
\end{figure}

\subsection{Effect of the Fuel-Saving Coefficient}
In this subsection, we aim at illustrating the effect of the fuel-saving coefficient $\beta$ on the behavior of the trucks. We perform all the simulations using the joint strategy fictitious play introduced in Algorithm~\ref{alg:1} with $p=0.4$ and $\lambda_t=3\times 10^{-2}$ for all $t\in\mathbb{N}_{0}$. Figure~\ref{figure3} illustrates the number of trucks for the learned Nash equilibrium at different time intervals for various choices of the coefficient~$\beta$. As we expect, when $\beta=0$, the trucks are reluctant to platoon (but instead stick to the time that favors them the most). However, as we increase the coefficient $\beta$, a higher number of trucks drive at the same time interval. Note that for $\beta= 4\times 10^{-3}$, all hundred trucks use the road during exactly one time interval (i.e, 8:00am to 8:15am).

\begin{figure}[!t]
\centering
\includegraphics[width=0.5\linewidth]{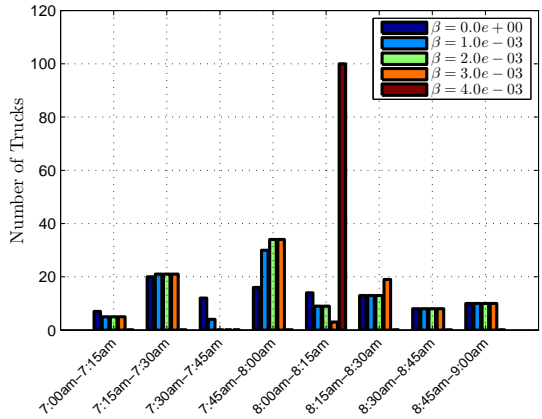}
\caption{\label{figure3} Number of the trucks in each time interval for various choices of the coefficient $\beta$.  }
\end{figure}

\subsection{Drivers Having Different Time Values}
In 2001, the consulting firm Inregia in Sweden, by the request of Swedish Institute for Transport and Communications Analysis, performed a survey to estimate the value of time for the road users in Stockholm~\cite{Inregiareport,engelson2006congestion}. This study showed that various groups of people value their time differently. According to the study, drivers valued time as $0.98$, $3.30$, and $0.19\,\mathrm{SEK/min}$ for work and school commuting trips, business trips, and other trips, respectively~\cite{Inregiareport,engelson2006congestion}. Let us include this effect in the introduced congestion game setup. Assume that in the utility of car $i\in\N$, we set the term 
$$
p_i^{\mathrm{c}}(z,x)=\delta_i^{-1}\hspace{-.04in}\left(a\beta \sum_{\ell=1}^{m_{z_i}(x)}g(\ell)\right)\hspace{-.05in},
$$
where $\delta_i> 0$ is the value of time for the driver of car~$i$. For work and school commuting trips, we scale the value of time to $\delta_i=1.00$. Therefore, we get $\delta_i=3.37$ and $\delta_i=0.19$ for business trips and other trips, respectively. Now, allow us to randomly distribute the cars into three groups of work and school trips, business trips, and other trips with probabilities $0.754$, $0.036$, $0.210$, respectively, as suggested in~\cite{engelson2006congestion}. Figure~\ref{figure4} shows the number of trucks in each time interval as a function of the iteration number in this case. Comparing with Figure~\ref{figure2}, we can clearly see that in this example, the difference in the value of time has not changed the behavior of trucks (certainly in the Nash equilibrium, but the transient response is different). Figure~\ref{figure_cars_value} shows the number of the cars in each time interval for the case where the drivers value their time differently subtracted by number of the cars in each time interval for the case where the drivers value their time equally. Clearly, the cars that value their time the most, or equivalently, the ones that are willing to pay higher congestion taxes (i.e., $\delta_i=1.00,3.37$), can move to the time interval where thirty trucks are traveling. However, the cars that do not value their time much (i.e., $\delta_i=0.19$) switch to a less expensive alternative.

\begin{figure}[!t]
\centering
\includegraphics[width=0.5\linewidth]{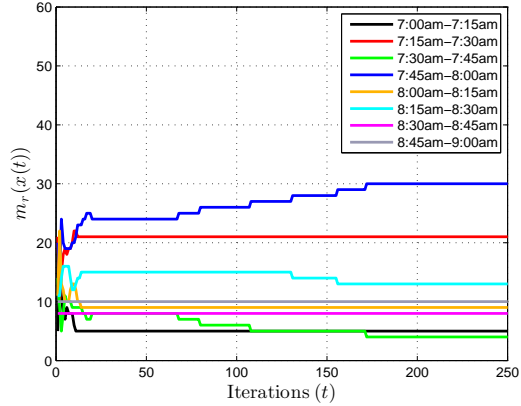}
\caption{\label{figure4} $m_{r}(x(t))$, $r\in\mathcal{R}$, versus the iteration number for the case where the drivers value their time differently.  }
\end{figure}

\begin{figure}[!t]
\centering
\includegraphics[width=0.5\linewidth]{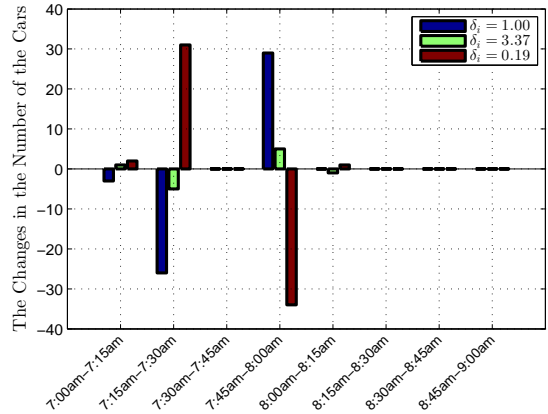}
\caption{\label{figure_cars_value} Number of the cars in each time interval for the case where the drivers value their time differently subtracted by number of the cars in each time interval for the case where their drivers value the time equally.  }
\end{figure}

\subsection{Trucks with and Without Platooning Equipment }
Few trucks are currently fitted with platooning equipments. In this subsection, we try to understand the influence of this matter on the properties of the learned Nash equilibrium. To illustrate the effect of trucks without platooning equipment, let us consider two types of trucks where the first type can indeed participate in platoons and the second type does not have the necessary equipments for doing so. We count the second type of trucks as ordinary cars since they do not benefit from traveling at the same time interval as the other trucks. Hence, $N$ shows the number of ordinary cars together with the trucks without platooning equipment and $M$ denotes the number of trucks that can potentially participate in forming the platoons. We fix $N+M=10000$. Figure~\ref{varying_number_of_trucks} illustrates the number of the trucks that have platooning equipment in each time interval for various ratios of $M/(M+N)$. Evidently, the number of the trucks (with platooning equipment) in most of the time intervals grows linearly with $M/(M+N)$ (as we expect since there are more trucks). However, some of the intervals, such as, 7:30am to 7:45am become less favorable (as they are highly congested) and the trucks in these intervals completely move to their neighboring intervals as $M/(M+N)$ increases.

\begin{figure}[!t]
\centering
\includegraphics[width=0.5\linewidth]{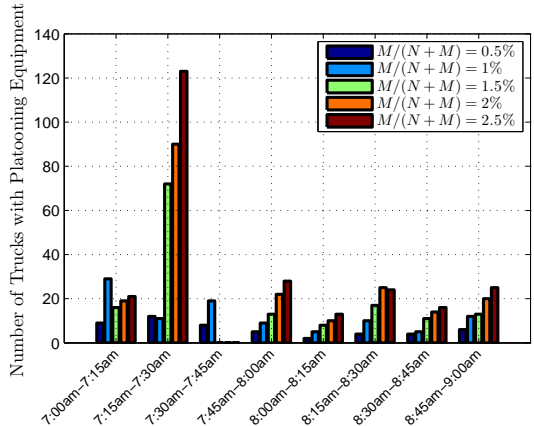}
\caption{\label{varying_number_of_trucks} Number of the vehicles in each time interval for the learned pure strategy Nash equilibrium for various choices of $M/(M+N)$. }
\end{figure}

\subsection{Announcing Congestion Taxes in Advance}
A drawback of the presented formulation is that the congestion taxes are dynamic and must be calculated (and enforced) instantly based on the number of the vehicles in each time interval. Although dynamic congestion taxing has been implemented on several occasions (e.g., San Diego I-15 High-Occupancy Toll Lanes in which the tolls vary dynamically with the level of congestion~\cite{USHighway}), they proved to be controversial (or, cumbersome to understand for the drivers at the least). Therefore, one might consider the case in which the tolls for day $t+D$ are announced at the end of day $t$ for all $t\in\mathbb{N}_0$ (so that the drivers have time to digest this information and act accordingly). To simulate such a scenario, we note that the congestion tax $p_i^{\mathrm{c}}(t)$ that car $i\in\N$ must pay for using the road at time interval $z_i(t)\in\mathcal{R}$ on iteration $t\in\mathbb{N}_0$ is equal
$$
p_i^{\mathrm{c}}(t)=\left\{ 
\begin{array}{ll}
a\beta \sum_{\ell=1}^{m_{z_i(t)}(x(t-D))}g(\ell), & t>D, \\
0, & \mbox{otherwise.}
\end{array}
\right.
$$
Figure~\ref{figuredelay} illustrates the number of the vehicles for each time interval $n_{r}(x(t),z(t))$, $r\in\mathcal{R}$, versus the iteration number when the congestion tax is updated with a delay of $D=30$ days. Evidently, there are more oscillations in comparison to Figure~\ref{figure1}, however, the algorithm converges rapidly to a pure strategy Nash equilibrium.

\begin{figure}[!t]
\centering
\includegraphics[width=0.5\linewidth]{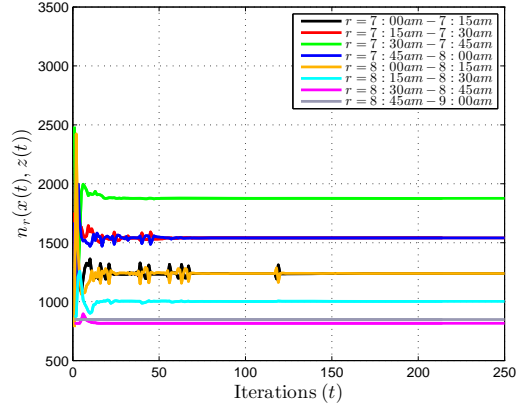}
\caption{\label{figuredelay} $n_{r}(x(t),z(t))$, $r\in\mathcal{R}$, versus the iteration number when the congestion tax is updated with a delay of $D=30$ days. }
\end{figure}

\subsection{Average Strategy Fictitious Play} \label{subsec:numericalexample:asfp}
In this subsection, we use the average strategy fictitious play with $\beta=10^{-3}$, $\lambda=3\times 10^{-2}$, and $p=0.4$. We also implement the platooning subsidy in Theorem~\ref{lem:2} with $v_0=85$. Figure~\ref{figure6} illustrates $n_{r}(x(t),z(t))$, $r\in\mathcal{R}$, versus the iteration number. The proposed algorithm clearly converges to a Nash equilibrium relatively fast.

\begin{figure}[!t]
\centering
\includegraphics[width=0.5\linewidth]{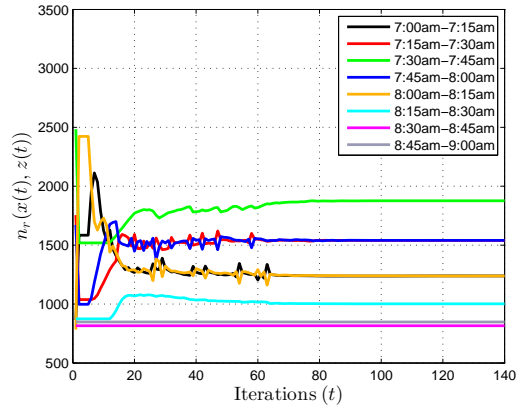}
\caption{\label{figure6} $n_{r}(x(t),z(t))$, $r\in\mathcal{R}$, versus the iteration number for $\beta=10^{-3}$ and $v_0=85$ when using the average strategy fictitious play in Algorithm~\ref{alg:2}. }
\end{figure}

\section{Conclusions and Future Work} \label{sec:conclusion}
We introduced a model for traffic flow on a specific road at various time intervals per day using an atomic congestion game with two types of agents (namely, cars and trucks). Cars only optimize their trade-off between using the road at the time they prefer, the average velocity of the traffic flow, and the congestion tax they are paying. However, trucks benefit from using the road at the same time as the other trucks. We motivated this extra utility using an increased possibility of platooning with the other trucks and as a result, saving fuel. We used congestion data from Stockholm to validate the linear relationship between the average velocity of commuting and the number of the vehicles that are using the road at that time. We devised appropriate tax or subsidy policies to create a potential game. Then, we used the joint strategy fictitious play and the average strategy fictitious play to learn a pure strategy Nash equilibrium of this game. We conducted a comprehensive simulation study to analyze the effect of different factors on the properties of the learned Nash equilibrium. As a future work, we can consider using mechanism design tools to enforce a socially optimal solution, such as, an optimal carbon emission profile, through appropriate congestion tax policy. Finally, in this paper, we did not consider the routing aspects of the problem. It would be of great interest in future research to combine the departure-time selection and the route selection problems in the context of understanding the platooning incentives.

\section*{Acknowledgement} 
The authors would like to thank Wilco Burghout for kindly providing the traffic data from the E4 highway in Stockholm. They would also like to thank Lihua Xie and Nan Xiao for initial discussions on the problem considered in this paper.

\bibliographystyle{ieeetr}
\bibliography{compile_new}

\end{document}